\documentclass[aps,superscriptaddress,tightenlines,10pt,twocolumn,nofootinbib]{revtex4}

\usepackage{amsmath,amssymb}
\usepackage{amsthm}
\usepackage{amsfonts}
\usepackage{color}
\usepackage{mathtools}
\usepackage{enumerate}
\usepackage{overpic}
\usepackage{bbm}

\usepackage[ colorlinks = true,
             linkcolor = blue,
             urlcolor  = blue,
             citecolor = red,
             anchorcolor = green,
]{hyperref}

\newtheorem{lemma}{Lemma}

\newtheorem{theorem}[lemma]{Theorem}

\DeclareMathOperator{\tr}{Tr}

\newcommand{\eps}{\varepsilon}

\newcommand*{\C}{\mathbb{C}}
\newcommand{\norm}[1]{\left\| #1 \right\|}

\newcommand{\bra}[2]{\mathinner{\langle #2|}_{#1}}
\newcommand{\ket}[2]{\mathinner{|#2\rangle}_{\hspace{-0.1em} #1}}

\def\bra #1{\langle #1\vert}
\def\ket #1{\vert #1\rangle}
\def\braket #1#2{\langle #1 \vert #2\rangle}

\newcommand\Scp[2]{\ensuremath{\, \langle #1 \,\vert #2 \,\rangle}}

\begin{document}

\title{Thermal States as Convex Combinations of Matrix Product States} 

\author{Mario Berta}
\affiliation{Department of Computing, Imperial College London, London, United Kingdom}

\author{Fernando G.~S.~L.~Brand\~ao}
\affiliation{Institute for Quantum Information and Matter, California Institute of Technology, Pasadena, California, USA}

\author{Jutho Haegeman}
\affiliation{Department of Physics and Astronomy, Ghent University, Ghent, Belgium}

\author{Volkher B.~Scholz}
\affiliation{Institute for Theoretical Physics, ETH Zurich, Z\"urich, Switzerland}
\affiliation{Department of Physics and Astronomy, Ghent University, Ghent, Belgium}

\author{Frank Verstraete}
\affiliation{Department of Physics and Astronomy, Ghent University, Ghent, Belgium}
\affiliation{Department of Physics, University of Vienna, Vienna, Austria}

\begin{abstract}
We study thermal states of strongly interacting quantum spin chains and prove that those can be represented in terms of convex combinations of matrix product states. Apart from revealing new features of the entanglement structure of Gibbs states our results provide a theoretical justification for the use of White's algorithm of minimally entangled typical thermal states. Furthermore, we shed new light on time dependent matrix product state algorithms which yield hydrodynamical descriptions of the underlying dynamics. 
\end{abstract}

\maketitle

%
\noindent\emph{Introduction}---%
The theory of entanglement and tensor networks have provided a novel language for describing strongly correlated quantum many body systems. This has led to a deeper understanding of the properties of topological phases of matter \cite{perez2008string,pollmann2010entanglement,chen2011classification,schuch2011classifying}, to novel computational algorithms for simulating such spin systems \cite{Verstraete_2008,schollwock2011density} and to rigorous proofs that ground states of gapped one-dimensional quantum spin systems can be represented \cite{Verstraete:2006ir,Hastings_2007} and simulated \cite{landau2013polynomial} using matrix product states as those ground states satisfy the area law for the entanglement theory. 

In this paper, we are concerned with representing thermal states of quantum spin systems. It has been proven that such Gibbs states satisfy the area law for the mutual information \cite{Wolf08}, that matrix product operators \cite{Verstraete_2004,zwolak2004mixed,Pirvu_2010} provide a faithful approximation \cite{Hastings_2006,molnar2015approximating} to Gibbs states, and efficient algorithms for finding those operators have been formulated \cite{Verstraete_2004,karrasch2012finite,Barthel}. Conceptually, those algorithms suffer from a major drawback: no distinction is made between the "classical" and "quantum" correlations. Classical correlations should be dealt with by using Monte Carlo sampling techniques, and one should not waste a large "bond dimension" to those fluctuations. Furthermore, all algorithms dealing with matrix product operators either deal with purifications, which can potentially lead to a huge increase in bond dimension \cite{de2013purifications}, or cannot assure positivity of the matrix product density operator. Those algorithms also become inefficient for low temperatures, as when working with matrix product operator descriptions of pure states the bond dimension is squared. Those problems can be cured by invoking mixtures of pure matrix product states, and this is the main topic of this paper. 

Our main result states that thermal states of one-dimensional local Hamiltonians can be approximately written as such as convex combination of Matrix product states, all of which have a bounded bond dimension. Unlike the case of ground states~\cite{Verstraete:2006ir,Hastings_2007}, the Hamiltonians do not need to be gapped --- we only require a uniform bound on the interaction strength. The proof relies on recent results by one of the authors concerning the Markov structure of Gibbs states~\cite{kato16}. We illustrate this by providing arguments for using Matrix product states as subroutines in algorithms dealing with thermal states of quantum spin systems: first, the METTS algorithm of White~\cite{White_2009,Stoudenmire_2010} yielding an approximation of Gibbs states using DMRG techniques, and second, for the quantum thermalization algorithm of Leviatan {\it et al.}~revealing hydrodynamical properties of quantum spin chains~\cite{2017arXiv170208894L}.

Our paper is organized as follows. We first review the definition of Matrix product states and state of main result. After sketching the proof (the details can be found in the appendices), we study the two mentioned numerical algorithms, starting with the static case which is followed by the dynamical case. We end by discussing further applications of our results. 

%
\noindent\emph{Matrix Product states}---%
A prominent example of quantum states representable by a tensor network are Matrix product states (MPS)~\cite{fannesnachtergaelewerner1992,Verstraete:2006ir,Wolf:2007wt}. They form a sub-manifold~\cite{Haegeman_2014} $M_{MPS}^D \subset \left(\C^d\right)^{\otimes n}$ of the state space of a one-dimensional quantum lattice system with $n$ sites and a finite $d$-dimensional local Hilbert space on every site
\begin{align}\label{eq:mpsdef}
\ket{\psi[A_{i_1},\ldots, A_{i_n}]} = \sum_{i_1,\ldots,i_n} \tr[A_{i_1} \cdots A_{i_n}] \,\ket{i_1}\ldots\ket{i_n},
\end{align}
where here and henceforth we assume periodic boundary conditions~\footnote{Our main result (Thm.~\ref{thm:main}) can be extended to open boundary conditions rather trivially but we choose to work with periodic boundary condition to keep the notation simple.}, and the $A_{i_j}$, $j = 1,\ldots, n$ are $D \times D$ dimensional matrices. The parameter $D$ is called the bond dimension and models the amount of entanglement in the state. Matrix product states with low bond dimension satisfy the area law of entanglement and have been proven to capture the ground state physics of one-dimensional local gapped Hamiltonians\,---\,exactly due to the rapid decay of the entanglement spectrum in these systems~\cite{Hastings_2007}. In addition, MPS allow for an efficient computation of expectation values, independently of the system size. Various algorithms exists which find the best approximate state within the sub-manifold $M_{MPS}^D$ for the ground state of a local Hamiltonian, either variationally using the density matrix renormalization group (DMRG)~\cite{RevModPhys.77.259} or by simulating imaginary time evolution using time-evolving block decimation (TEBD)~\cite{PhysRevLett.93.040502} or the time-dependent variational principle (TDVP)~\cite{PhysRevLett.107.070601}. In this letter, we are interested in mixed states which can be written as convex combinations of Matrix product states of a fixed bond dimension,
\begin{align}\label{eq:convmpsdef}
\rho[\mu] = \int d\mu(A_{i_1},\ldots) \ket{\psi[A_{i_1},\ldots]}\bra{\psi[A_{i_1},\ldots]}.
\end{align}
Here $\mu(A_{i_1},\ldots)$ denotes some probability measure on the manifold $M_{MPS}^D$ and is otherwise arbitrary. We note that if we can efficiently sample from this distribution, then we can also efficiently compute expectation values of local observables, as in the case of Matrix product states.

%
\noindent\emph{Main Result}---%
Sharing many properties with the class of pure Matrix product states which model the ground state physics of gapped local Hamiltonians, we expect that convex combinations of Matrix product states also possess physical significance. And indeed we find that they approximate thermal states
\begin{align}
\rho_{H,T} = \frac{1}{Z}\,\exp(-H/T)
\end{align}
of local Hamiltonians $H = \sum_{i} h_i$, where $Z = \tr[\exp(-H/T)]$ is the partition function ($T$:~temperature). Here local means that each term $h_i$ only acts on a finite number of neighbouring sites. In addition, we assume that the interaction terms all satisfy a unique upper bound on their interaction strength, $\norm{h_i}_{\infty} \leq C$, where $\norm{.}_\infty$ denotes the operator norm.

\begin{theorem}\label{thm:main}
Let $H$ be a local one-dimensional Hamiltonian such that it interaction terms possess a uniform upper bound on their interaction strength. Then, for any temperature $T$ and any $\eps>0$ there exists a bond dimension $D$ and a probability distribution $\mu_\eps$ on the manifold $M_{MPS}^D$ of Matrix product states with bond dimension $D$ such that the associated convex combination of Matrix product states $\rho[\mu]$ is $\eps$-close to the thermal state at temperature $T$:
\begin{align}\label{eq:result}
\norm{\rho_{H,T}-\rho[\mu_\eps]}_{1}\leq\eps.
\end{align}
Here, the $\norm{.}_1$ denotes the trace-norm. The bond dimension $D$ scales quasi-polynomially in the system size and $\eps^{-1}$, and doubly exponential in the inverse temperature $T^{-1}$.
\end{theorem}

Note that such an approximation is trivial for a bond dimension $D$ scaling exponentially in system size, whereas we show quasi-polynomial scaling. We emphasise that a bound on the trace-norm is in particular sufficient to guarantee a bound on the error in observables\,---\,but probably not required. Hence, numerical simulations might observe a faster convergence, especially in local observables. In the following, we present the main ideas which go into the proof of Thm~\ref{thm:main}, and refer the reader interested in the exact quantitative bounds and mathematical details to the appendices.

To start with, we mention that the entanglement structure of Gibbs states with finite correlation length is known to fulfil an area law for the correlation measure quantum mutual information~\cite{Wolf08}
\begin{align}
I(A:B)_\rho:=H(A)_\rho+H(B)_\rho-H(AB)_\rho,
\end{align}
where $H(A)_\rho:=-\mathrm{Tr}\left[\rho_A\log\rho_A\right]$ denotes the von Neumann entropy. Now, to prove Thm.~\ref{thm:main} it would be sufficient to extend that to an area law for another correlation measure: the so-called entanglement of formation~\cite{Bennett96}
\begin{align}
E_F(A:B)_\rho:=\inf\sum_i p_iH(A)_{\rho^i},
\end{align}
where the infimum is over all pure state decompositions $\rho_{AB}=\sum_ip_i\ket{\rho^i}\bra{\rho^i}_{AB}$. However, since the mutual information is neither a convex nor a concave function in the state it generally behaves very differently than the entanglement of formation. In particular, there exist quantum states with $E_F(A:B)_\rho \gg I(A:B)_\rho$~\cite{hayden06} and hence this argument cannot be used to prove the desired statement.


Another way to proceed for $L=\alpha\beta\gamma$ would be to focus on the quantum conditional mutual information, $I(\alpha:\gamma|\beta)_\rho = H(\alpha\beta)_\rho + H(\beta\gamma)_\rho - H(\beta)_\rho - H(\alpha\beta\gamma)_\rho$, where $\beta$ connects the regions $\alpha$ and $\beta$ (see Fig.~\ref{fig:Markov_chain}). If now $I(\alpha:\gamma|\beta)_\rho$ would be decaying exponentially in the size of $\beta$, then recent results in quantum information theory~\cite{Fawzi2015,arXiv:1509.07127,Sutter2017} show that we could recover the state on $\alpha\beta\gamma$ from the one on just $\alpha\beta$ using a completely positive trace-preserving map acting only on $\beta$ and having only a few Kraus operators. As already noted in~\cite{PhysRevB.94.155125,2016arXiv160907877B}, this would allow us to repeat that argument and obtain a representation of the global state using a sequence of quasi-local maps --- also called a local Markov chain structure.

This in turn will be the entry point of our arguments, as discussed in the following. Unfortunately, while we expect that for many interesting physical systems the quantum conditional mutual information behaves as expected, a general statement for thermal states of local Hamiltonians is not known. However, the local Markov chain structure of Gibbs states of one-dimensional local Hamiltonians still holds approximately~\cite{kato16}.

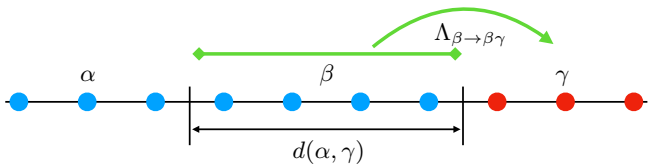
\begin{figure}[t]
\begin{overpic}[width=0.5\textwidth]{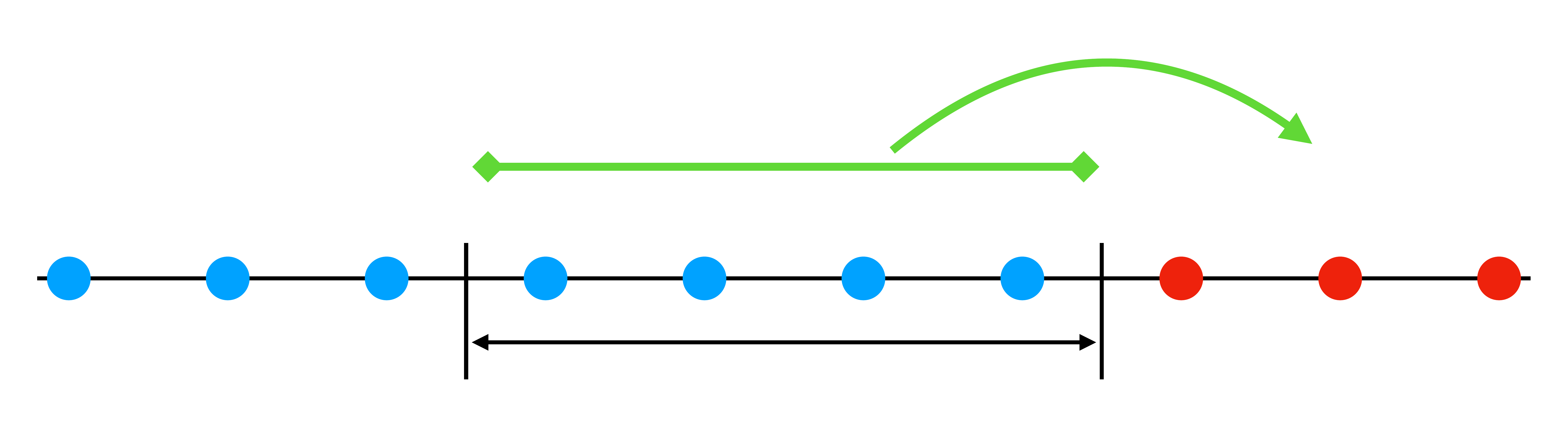}
\put(13.5,13.5){$\alpha$}
\put(49,13.5){$\beta$}
\put(84,13.5){$\gamma$}
\put(66,19.5){$\Lambda_{\beta\to\beta\gamma}$}
\put(45,2){$d(\alpha,\gamma)$}
\end{overpic}
\caption{Markov chain structure of Gibbs states: for every tripartite split of the lattice $\alpha\beta\gamma$, there exists a local quantum channel $\Lambda_{\beta\to\beta\gamma}$ only acting on the region $\beta$ and approximately recovering the global Gibbs state from the reduced state on the region $\alpha\beta$.}
\label{fig:Markov_chain}
\end{figure}

\begin{lemma}[Markov chain structure of Gibbs states]\label{lem:markov}
Let $H$ be a local one-dimensional Hamiltonian. Then, for any tripartite split $L:=\alpha\beta\gamma$ as depicted in Fig.~\ref{fig:Markov_chain}, there exists a local quantum channel $\Lambda_{\beta\to\beta\gamma}$ such that
\begin{align}\label{eq:markov}
\left\|\rho_{H,T}-(\mathcal{I}_\alpha\otimes\Lambda_{\beta\to\beta\gamma})(\rho_{\alpha\beta})\right\|_1\leq\exp\left(-q(T)\sqrt{d(\alpha,\gamma)}\right),
\end{align}
where $d(\alpha,\gamma)$ measures the distance in system size between the regions $\alpha$ and $\gamma$, and $q(T)$ denotes some temperature dependent constant.
\end{lemma}

As shown in~\cite{kato16} the intuition behind the recovery channel in Lem.~\ref{lem:markov} is as follows. First one uses the quantum belief propagation equations from~\cite{Hastings07} to obtain the following approximate decomposition of the Gibbs state
\begin{align}
&\exp\left(-H_{\alpha\beta\gamma}/T\right)\notag\\
&\approx P_{\beta,t}\Big(\exp\left(-H_{\alpha\beta_{L}}/T\right)\otimes\exp\left(-H_{\beta_{R}\gamma}/T\right)\Big)P_{\beta,t}^{\cal y}\,,\label{eq:fernando_intuition}
\end{align}
where $\beta_L\beta_R$ are regions each composed of half the qubits of $\beta$ and $P_{\beta,t}$ is a local operator of size $t:=d(\alpha,\gamma)/2$ centred in the cut between $\beta_L$ and $\beta_R$ with $\left\|P_{\beta,t}\right\|_\infty\leq2^{O(1/T)}$. The approximation error in Eq.~\eqref{eq:fernando_intuition} is exponentially small in $t$. From this decomposition one can define the completely positive and trace non-increasing map
\begin{align}
&K_{\beta\to\beta\gamma}(\cdot)\notag\\
&:=P_{\beta,t}\left(\mathrm{Tr}_{\beta_R}\left[P_{\beta,t}^{-1}(\cdot)\left(P_{\beta,t}^{-1}\right)^{\cal y}\right]\otimes\rho_{H_{\beta_R\gamma}}\right)P_{\beta,t}^{\cal y}\,,
\end{align}
and show that $\rho_{\alpha\beta\gamma}\approx\left(\mathcal{I}_A\otimes K_{\beta\to\beta\gamma}\right)(\rho_{\alpha\beta})$ up to an error $2^{-O(t)}$. This almost achieves our goal, however, the map $K_{B\to BC}$ is not trace preserving. To make it trace preserving one uses a repeat-until-success strategy by probabilistically implementing $K_{\beta\to\beta\gamma}$: if it succeeds the state is recovered and if it fails one traces out the region it was applied to and a shield region just next to it. Now since the Gibbs state has an exponential decay of correlations~\cite{araki69} and since the map is applied with probability $2^{-O(1/T)}$ independent of the system size one obtains a good approximation of the reduced state of the Gibbs state (on a slightly smaller length). One can then repeat the procedure until it is successful and this increases the total error to $2^{-O(\sqrt{t})}$ (see~\cite{kato16} for more details).

With the help of Lem.~\ref{lem:markov} we are now ready to give a proof of our main result.

\begin{proof}[Proof of Thm.~\ref{thm:main}]
Our argument is based on three steps:
\begin{enumerate}[(i)]
\item Using Lem.~\ref{lem:markov} we construct a Matrix product state $\ket{\Psi_{D,\eps}}$ with bond dimension $D$ quasi-polynomial in the system size $n$ and $\eps^{-1}$.
\item We show that $\ket{\Psi_{D,\eps}}$ is the purification of a convex combination of Matrix product states with bond dimension $D$ as in step (i) -- denoted by $\rho[\mu_\eps]$.
\item We show that $\rho[\mu_\eps]$ is close to the Gibbs state $\rho_{H,T}$.
\end{enumerate}
For step (i) we split the lattice into three consecutive regions (as depicted in Fig.~\ref{fig:main_result})
\begin{align}
L:=A_1B_1C_1\;A_2B_2C_2\;\cdots\;A_IB_IC_I
\end{align}
of dimension $|A_i|=|B_i|=2^l$ and $|C_i|=2^{l\cdot5\xi}$ where $\xi\leq\exp(c/T)$ denotes the correlation length of the Gibbs state~\cite{araki69}. We then prepare a purification $\ket{\rho^i}_{\bar{A}_iA_i\bar{B}_iB_i}$ of all the reduced states on $A_iB_i$ of the Gibbs state $\rho_{H,T}$ and fill in the missing $C_i$ pieces together with their purifications $\bar{C}_i$ by making use of the Markov chain structure of the Gibbs state (Lem.~\ref{lem:markov}). In more detail, we apply Lem.~\ref{lem:markov} to the full lattice $L$ with the decomposition
\begin{align}
\text{$\gamma_i:=C_i$, $\beta_i:=B_iA_{i+1}$, and $\alpha_i:=L/\left(\beta_i\gamma_i\right)$,}
\end{align}
leading to quantum channels $\Lambda^i_{B_iA_{i+1}\to B_iA_{i+1}C_i}$ with the approximation property as in Eq.~\eqref{eq:markov}. We then apply all the corresponding (minimal) Stinespring dilations of these channels leading to
\begin{align}
\ket{\Psi_{D,\eps}}:=\bigotimes_{i=1}^IV^i_{B_iA_{i+1}\to B_iA_{i+1}C_i\hat{B}_i\hat{A}_{i+1}\hat{C}_i}\ket{\rho^i}_{\bar{A}_iA_i\bar{B}_iB_i}\,,
\end{align}
where $|\hat{B}_i|=|B_i|^2$, $|\hat{A}_{i+1}|=|A_{i+1}|^2$, and $|\hat{C}_i|=|C_i|$. It is straightforward to check that the resulting global pure state $\ket{\Psi_{D,\eps}}$ becomes a Matrix product state with bound dimension upper bounded by $D\leq2^{l(8+10\xi)}$ and choosing $l=\log^2\left(n/\varepsilon\right)$ establishes step (i).

For step (ii) we arrive at $\rho[\mu_\eps]$ from $\ket{\Psi_{D,\eps}}$ by tracing out the purifying registers $\bar{A}_{i+1}\bar{B}_i$ as well as the Stinespring registers $\hat{B}_i\hat{A}_{i+1}\hat{C}_i$. By the monotonicity of the Schmidt rank under stochastic local operations assisted by classical communication (SLOCC)~\cite{vidal99} this exactly creates a convex combination of matrix product states with bond dimension upper bounded by $D$ from step (i).

Finally, step (iii) is deduced from the approximate quantum Markov structure as in Eq.~\eqref{eq:markov} together with a telescoping sum argument as well as some non-lockability bound from~\cite[Lem.~20]{horodecki14}. We refer to the appendices for a derivation of the exact error bounds.
\end{proof}

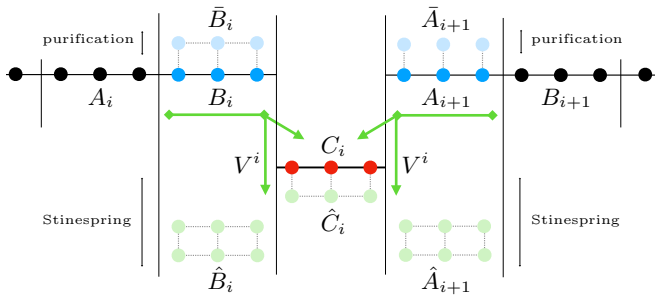
\begin{figure}[t]
\begin{overpic}[width=0.48\textwidth]{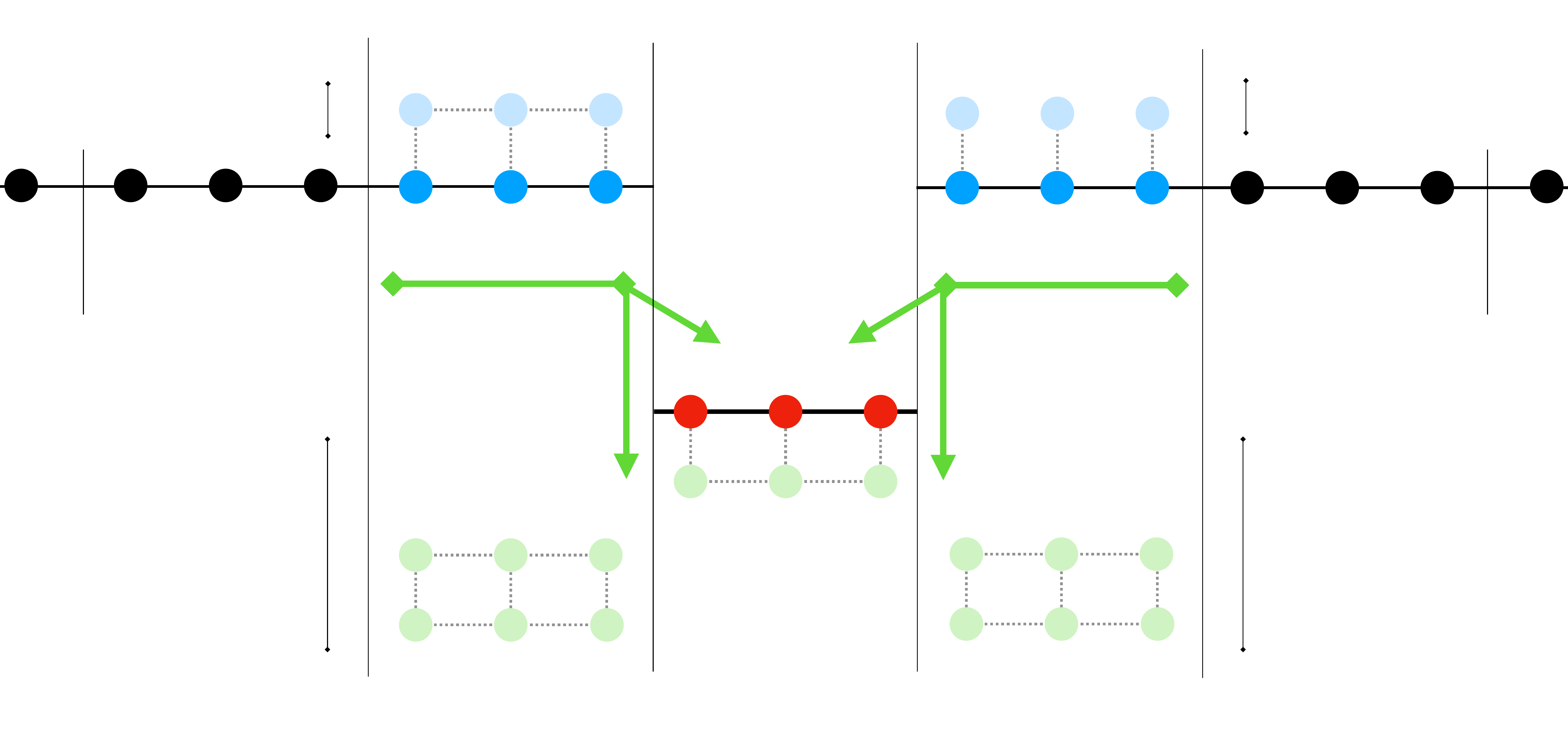}
\put(12.5,30.5){$A_i$}
\put(31,30.5){$B_i$}
\put(31,42.5){$\bar{B}_i$}
\put(31,2){$\hat{B}_i$}
\put(64,30.5){$A_{i+1}$}
\put(64,42.5){$\bar{A}_{i+1}$}
\put(64,2){$\hat{A}_{i+1}$}
\put(82.5,30.5){$B_{i+1}$}
\put(48.5,23){$C_i$}
\put(48.5,11){$\hat{C}_i$}
\put(35,20){$V^i$}
\put(61,20){$V^i$}
\put(81,40){\tiny purification}
\put(5.5,40){\tiny purification}
\put(81,12){\tiny Stinespring}
\put(5.5,12){\tiny Stinespring}
\end{overpic}
\caption{Proof of our main result (Thm.~\ref{thm:main}): starting from the purifications of the reduced states on $A_iB_i$ of the Gibbs state (blue region) we create the $C_i$ pieces together with their purifications (red region) from applying the Stinespring dilations $V^i$ of the local Markov channels $\Lambda^i$ from Lem.~\ref{lem:markov} (green region). We then get the desired convex combination of Matrix product states on the main lattice $A_iB_iC_i$ by tracing out all the purifying registers $\bar{A}_{i+1}\bar{B}_i$ and the Stinespring registers $\hat{B}_i\hat{A}_{i+1}\hat{C}_i$.}
\label{fig:main_result}
\end{figure}

%
\noindent\emph{Gibbs states and the METTS algorithm}---%
An algorithm to construct thermal ensembles of MPS is already available in the literature. It was introduced by White and goes by the name of the Minimally entangled typical thermal states (METTS) algorithm. Our theorem provides a theoretical justification for this algorithm.   Like in the MPO case, we can start from the trivial infinite temperature state and evolve it with imaginary time $\beta/2$ on both sides. However, instead of keeping the exact representation, we can represent the infinite temperature state by uniformly sampling over the basis of product states $\ket{\vec{i}} = \ket{i_1}\cdots\ket{i_n}$, and capture the action of the evolution $\exp(-\beta/2 H)$ on $\ket{\vec{i}}$ by promoting it to an MPS. We thus obtain the representation
\begin{align}\label{eq:metts}
 \frac{1}{Z} e^{-\beta H} =  \sum_{\vec{i}} p(\vec{i}) \ket{\phi_{T,\vec{i}}}\bra{\phi_{T,\vec{i}}},
\end{align}
with $p(\vec{i}) = Z^{-1}\Scp{\vec{i}}{\exp(-\beta H)\vec{i}}$ being the probability to sample the minimally entangled typical thermal state METTS $\ket{\phi_{T,\vec{i}}} = p(\vec{i})^{-1/2}\,\exp(-\beta H/2)\ket{\vec{i}}$. To accurately approximate the probability distribution $p(\vec{i})$, a Markov process was defined using standard time evolution algorithms for MPS, and as far as the steps are small and the entanglement growth is limited, this process satisfies an approximate version of detailed balance and typically leads to a good approximation of the equilibrium state. Note however that there is no guarantee that this algorithm converges to the true Gibbs state, as otherwise we would be able to simulate systems that are believed to be intractable (see, e.g., \cite{Aharonov2009}). Our theorem shows that a representation in terms of a mixture of MPS exists, not that it is easy to find this mixture.

%
\noindent\emph{Hydrodynamics and continuous time updates}---%
So far, we only considered the static case, that is, thermal equilibrium. But now consider a one-dimensional physical system which is first in thermal equilibrium and then, at time zero, subject to a quench such as a local spin flip at the origin, given by a unitary $U$. The system will eventually converge again into thermal equilibrium, albeit at a different temperature due to the injection of energy at time zero. Our theorem shows that at both ends of the time evolution, the state can be represented as a mixture of MPS.  It is hence natural to assume that the time dependent density matrix can at all times be represented as a mixture of matrix product states, as the "classical" entropy suppresses the quantum correlations. We are hence interested in developing an algorithm for doing time  evolution with mixtures of MPS.  


Recently, starting from~\cite{PhysRevLett.107.070601} Leviatan and co-workers~\cite{2017arXiv170208894L} suggested an algorithm for time-evolving quantum systems at infinite temperature. Their algorithm commences by sampling Matrix product states from a given distribution. In~\cite{2017arXiv170208894L}, mostly the uniform distribution was considered, but given our thoughts it is of course very plausible to start with the state given by the METTS algorithm as this is expected to approximate the thermal state before the spin flip. Next, we apply the local spin flip, which is easily implemented on each Matrix product state. Indeed, acting with $U$ on a MPS produces just another MPS. Hence acting with the local spin flip on the output of the METTS algorithm gives rise to another convex combination of MPS,
\begin{align}\label{eq:mettsspinflip}
\sum_{\vec{i}}\,q(\vec{i})\,\ket{\psi_{T,\vec{i}}}\bra{\psi_{T,\vec{i}}} &\to \sum_{\vec{i}}\,q(\vec{i})\,U^{\dagger}\ket{\psi_{T,\vec{i}}}\bra{\psi_{T,\vec{i}}} U \\
&= \sum_{\vec{i}}\,q(\vec{i})\,\ket{\psi^\prime_{T,\vec{i}}}\bra{\psi^\prime_{T,\vec{i}}},
\end{align}
with $\ket{\psi^\prime_{T,\vec{i}}} = U^{\dagger}\ket{\psi_{T,\vec{i}}}$. The system is then subject to the time evolution under the Hamiltonian $H$, $\ket{\psi^\prime_{T,\vec{i}}} \to \exp(-itH)\ket{\psi^\prime_{T,\vec{i}}}$. In general, this action will not produce an MPS, and the approximation would be even quite bad~\cite{PhysRevE.75.015202}. However, the key idea of~\cite{2017arXiv170208894L} is to ignore this difficulty and just project the state $\exp(-itH)\ket{\psi^\prime_{T,\vec{i}}}$ back onto the manifold $M_{MPS}^D$ using the method of TDVP applied to MPS~\cite{PhysRevLett.107.070601,PhysRevB.94.165116} such as to preserve all local constants of motion. This leads to an update rule which as a whole takes the initial convex combination of MPS  to a new mixture of MPS, which gives a completely new interpretation of their algorithm. Their algorithm can immediately be generalized as follows. 

\begin{widetext}
Considering an infinitesimal step of the evolution, we can transform the update rule to act directly on the probability distribution $\mu$, giving rise to a Fokker-Planck-like equation. Starting from the pure state evolution
$\mathrm{i}\frac{\partial \ket{\psi(t)}}{\partial t} = H \ket{\psi(t)}$,
we can systematically decompose the right hand side into terms acting parallel to the state, in the tangent space, in the double tangent space, and so on
\begin{align}
H \ket{\psi} = E \ket{\psi} + h^i \ket{\partial_i \psi} + h^{ij} \ket{\partial_i \partial_j \psi} + \ldots 
\end{align}
where the first three terms already give rise to an equality in case of a nearest neighbour Hamiltonian. Here, partial derivatives correspond to derivatives with respect to the variational parameters, i.e. every single entry of every single MPS tensor $A$. Furthermore, we have $E = \braket{\psi}{H|\psi}$, $h^i = g^{i,\bar{\jmath}} \braket{\partial_{\bar{\jmath}} \psi }{ H - E | \psi}$ where $g^{i,\bar{\jmath}}$ is the inverse of the metric $g_{\bar{\imath},j} = \braket{\partial_{\bar{\imath}}\psi}{\partial_j \psi}$ of the MPS manifold,
\begin{align}
h^{ij}= g^{ij,\bar{k}\bar{l}}  \left[\braket{\partial_{\bar{k}}\partial_{\bar{l}} \psi }{ H -E| \psi}-\braket{\partial_{\bar{k}}\psi}{H-E|\psi}\braket{\partial_{\bar{l}}\psi}{\psi}-\braket{\partial_{\bar{l}}\psi}{H-E|\psi}\braket{\partial_{\bar{k}}\psi}{\psi}\right]
\end{align}
with $g_{\bar{k}\bar{l},ij} = \braket{\partial_{\bar{k}}\partial_{\bar{l}} \psi }{ \partial_{i}\partial_{j} \psi}$, and so on. Partial integration in the convex combination of MPS leads to 
\begin{equation}
\begin{split}
	\frac{\partial \mu(\bar{A},A,t)}{\partial t} =& \mathrm{i} \left[\partial_i \big(h^i(\bar{A},A,t) \mu(\bar{A},A,t)\big)-\partial_{\bar{\imath}} \big(\overline{h}^{\bar{\imath}}(\bar{A},A,t) \mu(\bar{A},A,t)\big)\right]\\
	&-\mathrm{i} \left[\partial_i\partial_j \big(h^{i,j}(\bar{A},A,t) \mu(\bar{A},A,t)\big)-\partial_{\bar{\imath}}\partial_{\bar{\jmath}} \big(\overline{h}^{\bar{\imath},\bar{\jmath}}(\bar{A},A,t) \mu(\bar{A},A,t)\big)\right]
\end{split}
\end{equation}
The first order derivative term gives rise to a drift term of the probability distribution, and corresponds to the prediction of the TDVP (as applied in Ref.~\cite{2017arXiv170208894L}). In particular, this term will preserve a pure state. While the second derivative term might be interpreted as a diffusion term at first, one easily observes that the corresponding diffusion matrix is not positive definite. This is to be expected, as we would otherwise capture the exact evolution (for a nearest neighbour Hamiltonian) with our convex combination of MPS, even when starting from an initial pure state. In that case, the above equation will immediately give rise to negative values of the probability distribution $\mu$ and therefore to a sign problem when trying to sample. However, for initial distributions that are sufficiently broad, the above evolution might not destroy positivity right away and could yield an improvement over the pure drift case considered in Ref.~\cite{2017arXiv170208894L}.
\end{widetext}

%
\noindent\emph{Conclusions}---%
Using recently developed tools from quantum information theory, we proved that thermal states of local, one-dimensional Hamiltonians can be approximately written as a convex combination of Matrix product states. We employed this fact to reconsider two algorithms developed within condensed matter theory, and showed how our result provides a theoretical justification\,---\,much like the work of Hastings~\cite{Hastings_2007} does explain the success of DMRG. We believe that the description of thermal states in term of mixtures will provide the right framework for simulating quantum hydrodynamical effects. The main question left open by our work is if the bound dimension scaling can be improved from quasi-polynomial to polynomial in the system size $n$ and $\varepsilon^{-1}$. Note that this would be the case if we could obtain bounds on the conditional mutual information as strong as argued in~\cite{PhysRevB.94.155125,2016arXiv160907877B}. We emphasize that in contrast to Hasting's result for the MPS representation of ground states~\cite{Hastings_2007} our approximation of thermal states in Thm.~\ref{thm:main} does not assume the Hamiltonian to be gapped. Hence, correlations can be in general long-ranged and there is the possibility that Thm.~\ref{thm:main} can only be improved for, e.g., gapped Hamiltonians. Finally, it would be interesting to study extensions of our result to non-local Hamiltonians. The limitation of our current proof strategy is that we crucially employ a result of Araki~\cite{araki69}, which states that Gibbs states of one-dimensional local Hamiltonians have an exponential decay of correlations.

%
\noindent\emph{Acknowledgments}---%
We thank Isaac Kim for pointing out an error in an earlier version of this manuscript and Brian Swingle for discussions. JH acknowledges support from the ERC via Grant 715861 (ERQUAF). FV acknowledges funding from the ERC grant QUTE, and the sfb projects foqus and vicom. 


\begin{widetext}

\appendix

\section{Proof of Thm.~1}

Let $H=\sum_i h_i$ be a one-dimensional short-range Hamiltonian with $\|h_i\|_\infty\leq1$. For a temperature $T>0$, let
\begin{align}
\rho^{H,T}:=\frac{\exp\left(-H/T\right)}{\mathrm{Tr}\left[\exp\left(-H/T\right)\right]}
\end{align}
be the corresponding Gibbs state. We will make use of the following restatement of Lem.~2 from the main text -- which shows that Gibbs states of one-dimensional local Hamiltonians have an approximate local Markov chain structure~\cite[Thm.~1]{kato16}. 

\begin{lemma}[Local Markov chain structure of Gibbs states~\cite{kato16}]\label{lem:kato}
For every tripartite split of the lattice $\alpha\beta\gamma$, there exist a completely positive and trace preserving map $\Lambda_{\beta\to\beta\gamma}$ such that
\begin{align}
\left\|\rho^{H,T}_{\alpha\beta\gamma}-\left(\mathcal{I}_\alpha\otimes\Lambda_{\beta\to\beta\gamma}\right)\left(\rho^{H,T}_{\alpha\beta}\right)\right\|_1\leq\exp\left(-q(T)\sqrt{d(\alpha,\gamma)}\right)\,,
\end{align}
for any $d(\alpha,\gamma)\geq l_0$ and $q(T):=C\exp(-c/T)$ (where $0<l_0,C,c<100$ are universal constants and $d(\alpha,\gamma)$ quantifies the minimal distance in system size between $\alpha$ and $\gamma$).
\end{lemma}

More precisely we will make use of an adapted, purified version of~\cite[Cor.~4]{kato16}. This original corollary says that Gibbs states of one dimensional local Hamiltonians can be well-approximated by a depth-two circuit with each gate acting locally on $O(\log^2(n))$ qubits. The following is a more precise restatement of Thm.~1 from the main text that we seek to prove.

\begin{theorem}
For a $n$-qubit system with fixed temperature $T>0$ we can write
\begin{align}
\left\|\rho^{H,T}-\sum_{j=1}^{J}p_j|\varphi^j\rangle\langle\varphi^j|\right\|_1\leq\varepsilon\,,
\end{align}
where $p_j$ denotes a probability distribution and all the $|\varphi^j\rangle$ are matrix product states with bond dimension
\begin{align}\label{eq:bond_dimension}
D=2^{O\left(\log^2\left(n/\varepsilon\right)\right)}\,.
\end{align}
Furthermore, $J$ is exponential in $n$ and the temperature dependence of the bond dimension is $\exp(\exp(O(1/T)))$.
\end{theorem}

The full argument is shown here for maximal readability, where some steps similar to the proof of~\cite[Cor.~4]{kato16} are reproduced.

\begin{proof}
The proof is done in multiple steps:
\begin{enumerate}[(i)]
\item We construct one global matrix product state $\ket{\Psi_{D,\eps}}$ with bond dimension $D$ as in~\eqref{eq:bond_dimension}.
\item We show that this state $\ket{\Psi_{D,\eps}}$ is the purification of some state $\sum_{j=1}^{J}p_j|\varphi^j\rangle\langle\varphi^j|$ with all the $|\varphi^j\rangle$ matrix product states having bond dimension $D$ as in~\eqref{eq:bond_dimension}.
\item We show that this state $\sum_{j=1}^{J}p_j|\varphi^j\rangle\langle\varphi^j|$ is $\varepsilon$-close to the Gibbs state $\rho^{H,T}$ in trace distance.
\end{enumerate}

For (i) split the lattice into three consecutive regions $A_1B_1C_1A_2B_2C_2\cdots A_{I}B_{I}C_{I}$ of dimensions $|A_i|=|B_i|=2^l$ and $|C_i|=2^{l\cdot5\xi}$ with $\xi\leq\exp(c/T)$ the correlation length of the Gibbs state where $c$ is the constant from Lem.~\ref{lem:kato}~\footnote{Gibbs states of 1D local Hamiltonians have an exponential decay of correlations~\cite{araki69}.}. Set $\rho^{H,T}=:\rho_{A_1B_1C_1\cdots A_{I}B_{I}C_{I}}$. To construct the global state $\ket{\Psi_{D,\eps}}$ we first prepare a purification of all the marginals $\rho_{A_iB_i}$ of the Gibbs state: $|\zeta^i\rangle_{A_i\bar{A}_iB_i\bar{B}_i}$ with $|\bar{A}_i|=|A_i|$ and $|\bar{B}_i|=|B_i|$. Now the local Markov chain structure of Gibbs states (Lem.~\ref{lem:kato}) implies the existence of completely positive trace preserving maps $\Lambda^i_{B_iA_{i+1}\to B_iA_{i+1}C_i}$ such that~\footnote{For $\alpha_i=A_1B_1C_1A_{i-1}B_{i-1}C_{i-1}A_iB_{i+1}C_{i+1}\cdots A_{I}B_{I}C_{I}$, $\beta_i=B_iA_{1+1}$, and $\gamma_i=C_i$.}
\begin{align}
\left\|\Lambda^i_{B_iA_{i+1}\to B_iA_{i+1}C_i}\left(\rho_{A_1B_1C_1A_iB_i\cdots A_{I}B_{I}C_{I}}\right)-\rho_{A_1B_1C_1\cdots A_{I}B_{I}C_{I}}\right\|_1&\leq\exp\left(-q(T)\sqrt{\log\left|B_i\right|}\right)\\
&=\exp\left(-q(T)\sqrt{l}\right)\,.\label{eq:thm_inaction}
\end{align}
Denote Stinespring isometries of the $\Lambda^i_{B_iA_{i+1}\to B_iA_{i+1}C_i}$ by $V^i_{B_iA_{i+1}\to B_iA_{i+1}C_i\hat{D}_i}$, where we chose the dilation registers as $\hat{D}_i:=\hat{A}_i\hat{B}_{i+1}\hat{C}_i$ with $|\hat{A}_i|=|A_i|^2$, $|\hat{B}_{i+1}|=|B_{i+1}|^2$, and $|\hat{C}_i|=|C_i|$ (cf.~Fig.~2 from the main text)~\footnote{The dimension of the Stinespring dilation register is without loss of generality as small as input dimension times output dimension of the completely positive trace preserving map in question (see, e.g., \cite{kretschmann08}).}. We define
\begin{align}
\ket{\Psi_{D,\eps}}_{A_1\bar{A}_1B_1\bar{B}_1C_1\hat{D}_1\cdots A_{I}\bar{A}_{I}B_{I}\bar{B}_{I}C_{I}\hat{D}_{I}}:=\bigotimes_{i=1}^{I}V^i_{B_iA_{i+1}\to B_iA_{i+1}C_i\hat{D}_i}\left(\bigotimes_{i=1}^{I}|\zeta^i\rangle_{A_i\bar{A}_iB_i\bar{B}_i}\right)\,,
\end{align}
where the last $V^I_{B_I\to B_IC_I\hat{D}_I}$ is only acting on $B_I$. Since small Stinespring dilations imply small Schmidt-rank (Lem.~\ref{lem:schmidt}) this becomes a matrix product state with bond dimension~\footnote{In more detail, to apply Lem.~\ref{lem:schmidt} we need to extend the maps $\Lambda^i_{B_iA_{i+1}\to B_iA_{i+1}C_i}$ to maps $\Lambda^i_{B_iA_{i+1}C_i}$, e.g., by just first throwing the $C_i$-register away and then applying $\Lambda^i_{B_iA_{i+1}\to B_iA_{i+1}C_i}$. Using Lem.~\ref{lem:schmidt} with $\alpha_2=A_1\bar{A}_1B_1\bar{B}_1\cdots A_i\bar{A}_i\bar{B}_i$, $\alpha_1=B_iC_i$, $\beta_1=A_{i+1}$, and $\beta_2=\bar{A}_{i+1}B_{i+1}\bar{B}_{i+1}\cdots A_k\bar{A}_kB_k\bar{B}_k$ gives Eq.~\eqref{eq:dimension}.}
\begin{align}\label{eq:dimension}
\max_{i=1,\ldots,I}|A_i||B_i|\cdot|A_{i+1}||B_i|\cdot|A_{i+1}|^2|B_i|^2|C_i|^2\leq2^{l(8+10\xi)}\,.
\end{align}
Choosing $l=\log^2\left(n/\varepsilon\right)$ establishes step (i).

For (ii) we trace out the purifying systems $\bar{A}_1\bar{B}_1\cdots\bar{A}_{I}\bar{B}_{I}$ as well as the Stinespring dilation registers $\hat{D}_1\cdots\hat{D}_{I}$, leading to the state
\begin{align}\label{eq:close_state}
\bigotimes_{i=1}^{I}\Lambda^i_{B_iA_{i+1}\to B_iA_{i+1}C_i}\left(\bigotimes_{i=1}^{I}\rho_{A_iB_i}\right)=:\sum_{j=1}^{J}|\bar{\varphi}^j\rangle\langle\bar{\varphi}^j|_{A_1B_1C_1\cdots A_{I}B_{I}C_{I}}\,,
\end{align}
where each $|\bar{\varphi}^j\rangle\langle\bar{\varphi}^j|$ is generated by conditioning on a basis element of the traced out systems. We normalize to
\begin{align}
\text{$|\varphi^j\rangle\langle\varphi^j|:=|\bar{\varphi}^j\rangle\langle\bar{\varphi}^j|/p_j$ with $p_j:=\left\||\varphi^j\rangle\right\|_2^2$,}
\end{align}
and choose a local basis for the traced out systems $\bar{A}_1\bar{B}_1\hat{D}_1\cdots\bar{A}_{I}\bar{B}_{I}\hat{D}_{I}$ (see Fig.~2 from the main text). We can then use that the Schmidt-rank is monotone under local operations and post-selection -- so-called SLOCC monotonicity (see e.g., \cite{vidal99}) -- and hence conclude that all the $|\varphi^j\rangle$ are again matrix product states with bond dimension $D$ as in~\eqref{eq:bond_dimension}. $J$ becomes exponential in $n$.

For (iii) we need to show that the state in Eq.~\eqref{eq:close_state} is close to the Gibbs state $\rho_{A_1B_1C_1\cdots A_{I}B_{I}C_{I}}$ (the following part is a slightly simplified version of the proof of~\cite[Cor.~4]{kato16}). We first note that by the monotonicity of the trace distance under partial trace Eq.~\eqref{eq:thm_inaction} implies
\begin{align}\label{eq:thm_inaction2}
\left\|\Lambda^i_{B_iA_{i+1}\to B_iA_{i+1}C_i}\left(\rho_{A_1B_1C_1\cdots A_{i-1}B_{i-1}C_{i-1}A_iB_iA_{i+1}B_{i+1}}\right)-\rho_{A_1B_1C_1\cdots A_iB_iC_iA_{i+1}B_{i+1}}\right\|_1\leq\exp\left(-q(T)\sqrt{l}\right)\,.
\end{align}
We estimate
\begin{align}\label{eq:to_iterate}
&\left\|\bigotimes_{i=1}^{I}\Lambda^i_{B_iA_{i+1}\to B_iA_{i+1}C_i}\left(\bigotimes_{i=1}^{I}\rho_{A_iB_i}\right)-\rho_{A_1B_1C_1\cdots A_{I}B_{I}C_{I}}\right\|_1\notag\\
&\leq\left\|\Lambda^1_{B_1A_2\to B_1A_2C_1}\left(\rho_{A_1B_1}\otimes_{A_2B_2}\right)-\rho_{A_1B_1C_1A_2B_2}\right\|_1\notag\\
&\quad+\left\|\bigotimes_{i=2}^{I}\Lambda^i_{B_iA_{i+1}\to B_iA_{i+1}C_i}\left(\rho_{A_1B_1C_1A_2B_2}\bigotimes_{i=3}^{I}\rho_{A_iB_i}\right)-\rho_{A_1B_1C_1\cdots A_{I}B_{I}C_{I}}\right\|_1\,,
\end{align}
where we have used the triangle inequality for the trace distance and the monotonicity of the trace distance under partial trace (plus the sub-additivity with respect to tensor products). To bound the first term in Eq.~\eqref{eq:to_iterate} we use the estimate
\begin{align}
&\left\|\rho_{A_1B_1C_1\cdots A_{i-1}B_{i-1}A_iB_i}-\rho_{A_1B_1C_1\cdots A_{i-1}B_{i-1}}\otimes\rho_{A_iB_i}\right\|_1\notag\\
&\leq\Big(|A_i||B_i|\Big)^2\notag\\
&\quad\times\max_{\|X\|_\infty,\|Y\|_\infty\leq1}\left|\mathrm{Tr}\left[(\left(X_{A_1B_1C_1\cdots A_{i-1}B_{i-1}A_iB_i}\otimes Y_{A_iB_i}\right)\left(\rho_{A_1B_1C_1\cdots A_{i-1}B_{i-1}A_iB_i}-\rho_{A_1B_1C_1\cdots A_{i-1}B_{i-1}}\otimes\rho_{A_iB_i}\right)\right]\right|\\
&\leq2^{4l}\times2^{\frac{-\log\left|C_{i-1}\right|}{\xi}}\\
&\leq\exp(-l)\,,
\end{align}
where we used~\cite[Lem.~20]{horodecki14} for the first inequality and the exponential decay of correlations for the second inequality~\cite{araki69}. In particular, this then implies
\begin{align}
\left\|\Lambda^1_{B_1A_2\to B_1A_2C_1}\left(\rho_{A_1B_1}\otimes_{A_2B_2}\right)-\rho_{A_1B_1C_1A_2B_2}\right\|_1&\leq\left\|\Lambda^1_{B_1A_2\to B_1A_2C_1}\left(\rho_{A_1B_1A_2B_2}\right)-\rho_{A_1B_1C_1A_2B_2}\right\|_1+\exp(-l)\\
&\leq\exp\left(-q(T)\sqrt{l}\right)+\exp(-l)\,,
\end{align}
where we used Eq.~\eqref{eq:thm_inaction2} in the second step together with the monotonicity of the trace distance under partial trace. To estimate the second term in Eq.~\eqref{eq:to_iterate} we iterate the argument leading to
\begin{align}
\left\|\bigotimes_{i=1}^I\Lambda^i_{B_iA_{i+1}\to B_iA_{i+1}C_i}\left(\bigotimes_{i=1}^{I}\rho_{A_iB_i}\right)-\rho_{A_1B_1C_1\cdots A_IB_IC_I}\right\|_1&\leq I\cdot\left(\exp\left(-q(T)\sqrt{l}\right)+\exp(-l)\right)\\
&\leq I\cdot O\left(\varepsilon/n\right)\,,
\end{align}
since $l=\log^2\left(n/\varepsilon\right)$ as chosen before. We have $I\leq n$ and hence the claim follows.
\end{proof}


\section{Miscellaneous Lemmas}

\begin{lemma}[Small Stinespring dilations imply small Schmidt-rank]\label{lem:schmidt}
Let be $|\varphi\rangle_{\alpha_2\alpha_1\beta_1\beta_2}:=|\varphi\rangle_{\alpha_2\alpha_1}\otimes|\varphi\rangle_{\beta_1\beta_2}$ with Schmidt-rank upper bounded by $d$ for any cut through $\alpha_2\alpha_1$ or through $\beta_1\beta_2$. Moreover, let $\Lambda_{\alpha_1\beta_1}$ be a completely positive and trace preserving map with Kraus decomposition $\{K^i_{\alpha_1\beta_1}\}_{i=1}^{I}$. Then, there exists a Stinespring dilation isometry $V_{\alpha_1\beta_1\to\alpha_1\gamma\beta_1}$ of $\Lambda_{\alpha_1\beta_1}$ such that the Schmidt-rank of
\begin{align}
\ket{\Psi_{D,\eps}}_{\alpha_2\alpha_1\gamma\beta_1\beta_2}:=\left(1_{\alpha_2}\otimes V_{\alpha_1\beta_1\to\alpha_1\gamma\beta_1}\otimes1_{\beta_2}\right)|\varphi\rangle_{\alpha_2\alpha_1\beta_1\beta_2}
\end{align}
in any bipartite cut through $\alpha_2\alpha_1\gamma\beta_1\beta_2$ is upper bounded by $d|\alpha_1||\beta_1|\cdot I$.
\end{lemma}

\begin{proof}
The total system is given by $\alpha_2\alpha_1\gamma\beta_1\beta_2$ and we treat the different possible bipartite cuts separately:
\begin{enumerate}[(i)]
\item For any cut through $\alpha_2$ as well as for the cut $\alpha_2|\alpha_1\gamma\beta_1\beta_2$ the Schmidt-rank is upper bounded by $\log d$ since $V_{\alpha_1\beta_1\to\alpha_1\gamma\beta_1}$ does not act on $\alpha_2\beta_2$.

\item For the cut $\alpha_2\alpha_1|\gamma\beta_1\beta_2$ we choose
\begin{align}\label{eq:kraus}
V_{\alpha_1\beta_1\to\alpha_1\gamma\beta_1}(\cdot)=\sum_{i=1}^{I}K^i_{\alpha_1\beta_1}(\cdot)\otimes|i\rangle_\gamma\,.
\end{align}
Taking the operator Schmidt decomposition for each $K^i_{\alpha_1\beta_1}$ in the cut $\alpha_1|\beta_1$ gives
\begin{align}
K^i_{\alpha_1\beta_1}=\sum_{k_i=1}^{|\alpha_1||\beta_1|}K_{\alpha_1}^{k_i}\otimes K_{\beta_1}^{k_i}
\end{align}
and we find for the output state
\begin{align}\label{eq:gamma}
\ket{\Psi_{D,\eps}}_{\alpha_2\alpha_1|\gamma\beta_1\beta_2}=\sum_{i=1}^{I}\sum_{k_i=1}^{|\alpha_1||\beta_1|}\underbrace{\left(1_{\alpha_2}\otimes K_{\alpha_1}^{k_i}\right)|\varphi\rangle_{\alpha_2\alpha_1}}_{=:|\varphi^{ik_i}\rangle_{\alpha_2\alpha_1}}\otimes\underbrace{|i\rangle_\gamma\otimes\left(K_{\beta_1}^{k_i}\otimes1_{\beta_2}\right)|\varphi\rangle_{\beta_1\beta_2}}_{=:|\varphi^{ik_i}\rangle_{\gamma\beta_1\beta_2}}\,.
\end{align}

\item Cuts through $\gamma$ are already treated by Eq.~\eqref{eq:gamma}.

\item For any cut $\alpha_1=\alpha_1^a\alpha_1^b$ we take the Schmidt decomposition
\begin{align}
|\varphi\rangle_{\alpha_2\alpha_1}=\sum_{j=1}^d|\varphi^j\rangle_{\alpha_2\alpha_1^a}\otimes|\varphi^j\rangle_{\alpha_1^b}\,,
\end{align}
as well as the following operator Schmidt-decomposition of the Kraus operators in Eq.~\eqref{eq:kraus}:
\begin{align}
K^i_{\alpha_1\beta_1}=\sum_{k_i=1}^{|\alpha_1| |\beta_1|}K_{\alpha_1^a}^{k_i}\otimes K_{\alpha_1^b\beta_1}^{k_i}\,.
\end{align}
Hence, we find for the output state
\begin{align}
\ket{\Psi_{D,\eps}}_{\alpha_2\alpha_1^a|\alpha_1^b\gamma\beta_1\beta_2}=\sum_{i=1}^{I}\sum_{k_i=1}^{|\alpha_1||\beta_1|}\sum_{j=1}^d\underbrace{\Big(1_{\alpha_2}\otimes K_{\alpha_1^a}^{k_i}\Big)|\varphi^j\rangle_{\alpha_2\alpha_1^a}}_{=:|\varphi^{jik_i}\rangle_{\alpha_2\alpha_1^a}}\otimes\underbrace{\Big(K_{\alpha_1^b\beta_1}^{k_i}\otimes1_{\gamma\beta_2}\Big)\Big(|\varphi^j\rangle_{\alpha_1^b}\otimes|i\rangle_\gamma\otimes|\varphi\rangle_{\beta_1\beta_2}\Big)}_{=:|\varphi^{jik_i}\rangle_{\alpha_1^b\gamma\beta_1\beta_2}}\,.
\end{align}

\item Cuts through $\beta_2$, the cut $\alpha_2\alpha_1\gamma\beta_1|\beta_2$, as well as cuts of the form $\beta_1=\beta_1^a\beta_2^b$ follow by symmetry.
\end{enumerate}
We conclude the claim since we have treated all cuts.
\end{proof}

\end{widetext}

\bibliography{bibliography}


\end{document}